\documentclass[journal]{IEEEtran}

\usepackage{setspace}
\usepackage{lettrine}

\usepackage{graphics,
           psfrag,
           epsfig,
           amsthm,
           cite,
           amssymb,
           url,
           dsfont,
           subfigure,
           algorithm,
           algorithmic,
           balance,
           enumerate,
           color,
           setspace
}
\usepackage{amsmath}

\usepackage{epstopdf}

\newtheorem{definition}{Definition}

\newtheorem{lemma}{Lemma}

\newtheorem{theorem}{Theorem}
\newtheorem{corollary}{Corollary}

\newtheorem{remark}{Remark}

\newcommand{\eref}[1]{(\ref{#1})}
\newcommand{\sref}[1]{Section~\ref{#1}}
\newcommand{\appref}[1]{Appendix~\ref{#1}}
\newcommand{\fref}[1]{Figure~\ref{#1}}

\newcommand{\cref}[1]{Constraint~\ref{#1}}
\newcommand{\thref}[1]{Theorem~\ref{#1}}
\newcommand{\corref}[1]{Corollary~\ref{#1}}
\newcommand{\lref}[1]{Lemma~\ref{#1}}

\newcommand{\tref}[1]{Table~\ref{#1}}
\newcommand{\algref}[1]{Algorithm~\ref{#1}}

\newcommand{\ignore}[1]{}

\IEEEoverridecommandlockouts
\begin{document}

\title{Rate Aware Instantly Decodable Network Codes}

\author{
   \authorblockN{Ahmed Douik, \textit{Student Member, IEEE}, Sameh Sorour, \textit{Member, IEEE},\\ Tareq Y. Al-Naffouri, \textit{Member, IEEE}, and Mohamed-Slim Alouini, \textit{Fellow, IEEE}}
    
    {\thanks {A part of this paper is submitted to IEEE Global Telecommunications Conference Workshop (GLOBECOM' 2015), San Diego, California, USA.
    
Ahmed Douik is with the Department of Electrical Engineering, California Institute of Technology, Pasadena, CA 91125 USA (e-mail: ahmed.douik@caltech.edu).

Sameh Sorour is with the Electrical Engineering Department, King Fahd University of Petroleum and Minerals, Dhahran 31261, Saudi Arabia (e-mail: samehsorour@kfupm.edu.sa).

T. Y. Al-Naffouri is with King Abdullah University of Science and Technology, Thuwal 23955-6900, Saudi Arabia, and also with King Fahd University of Petroleum and Minerals, Dhahran 31261, Saudi Arabia (e-mail: tareq.alnaffouri@kaust.edu.sa).

M.-S. Alouini is with the Division of Computer, Electrical and Mathematical Sciences and Engineering, King Abdullah University of Science and Technology, Thuwal 23955-6900, Saudi Arabia (e-mail: slim.alouini@kaust.edu.sa).
}}
    }

\maketitle

\IEEEoverridecommandlockouts

\begin{abstract}
This paper addresses the problem of reducing the delivery time of data messages to cellular users using instantly decodable network coding (IDNC) with physical-layer rate awareness. While most of the existing literature on IDNC does not consider any physical layer complications and abstract the model as equally slotted time for all users, this paper proposes a cross-layer scheme that incorporates the different channel rates of the various users in the decision process of both the transmitted message combinations and the rates with which they are transmitted. The consideration of asymmetric rates for receivers reflects more practical application scenarios and introduces a new trade-off between the choice of coding combinations for various receivers and the broadcasting rate for achieving shorter completion time. The completion time minimization problem in such scenario is first shown to be intractable. The problem is, thus, approximated by reducing, at each transmission, the increase of an anticipated version of the completion time. The paper solves the problem by formulating it as a maximum weight clique problem over a newly designed rate aware IDNC (RA-IDNC) graph. The highest weight clique in the created graph being potentially not unique, the paper further suggests a multi-layer version of the proposed solution to improve the obtained results from the employed completion time approximation. Simulation results indicate that the cross-layer design largely outperforms the uncoded transmissions strategies and the classical IDNC scheme.
\end{abstract}

\begin{keywords}
Instantly decodable network coding, rate adaptation, completion time reduction, graph theory, maximum weight clique.
\end{keywords}

\section{Introduction} \label{sec:int}

\lettrine[lines=2]{I}{ntroduced} by Ahlswede \textit{et al.} in their seminal paper \cite{850663}, network coding (NC) has become a propitious paradigm for implementation in future wireless networks, providing fast and reliable real-time communications over fading channels, e.g., cellular, WiFi, and WiMAX. Based on the simple idea that the source and/or the intermediate nodes in a network can generate coded combinations of data messages \cite{6512065}, NC is a promising technique to improve throughput and to reduce delay in wireless networks.

Two categories of network coding can be distinguished in the literature named the Random Network Coding \cite{1228459,4015713} (RNC) and the Opportunistic Network Coding \cite{4937121,6145512,5191394} (ONC). RNC mixes messages using random and independent coefficients resulting in an optimality in reducing the number of transmissions and an ability to recover even without feedback. However, it is not suitable for real-time applications of interest in this paper as decoding can be performed only after the reception of the whole frame. On the other hand, ONC exploits the diversity of received and lost messages to generate the mixes online resulting in a lower decoding delay.

Recently, Instantly Decodable Network Coding (IDNC), an ONC subclass, captivated a significant number of works \cite{6655395,6766433,6620795,20112430,7037599,6809217,6882208} thanks to its fast decoding potential, essential for real-time applications, such as multimedia streaming \cite{6725590,5753573}. In IDNC, the coded message combinations need to be decoded at their reception instant and cannot be stored for future decoding. In other words, users that cannot decode a received message combination must discard it. To guarantee fast decoding, coded combinations in IDNC are encoded using binary XOR and are designed to offer instantaneous, low complexity message decoding. Despite the aforementioned instant decodability restriction, IDNC is still favorable in many applications requiring progressive and quick decoding, e.g., IPTV. Furthermore, IDNC is well suited for order-insensitive application requiring messages regardless of reception order such as roadside safety messages.

The problem of reducing the completion time in IDNC-based networks is the one of finding this optimal schedule of coded message combination so as to minimize the whole number of transmissions. Such optimization is very difficult even for erasure-free \cite{6775017} and offline coding \cite{4895447,5205612} scenarios, i.e., perfect knowledge of the future channel realizations. The problem is even more complicated in online \cite{5191398,1928762,4512941} coding scenarios as it becomes anti-causal due to the dynamic nature of the channel realizations and the dependence of the optimal solution of their effect.

The authors in \cite{6882208} formulate the optimal solution to the completion time minimization problem in IDNC-based networks as a shortest stochastic path (SSP). Despite the intractability of solving the SSP, the formulation allows to draw the theoretical guidelines for the policies that can minimize the completion delay. In \cite{5963123,6809217}, the authors extend their formulation to the limited and intermittent feedback scenarios. Reference \cite{6381028} investigate the minimization of the completion time in a time division duplex (TDD) IDNC-enabled network. In \cite{7037599}, the authors propose a new approach to solving the completion time problem based on a decoding delay control. The key idea is to approximate the completion time using a decoding delay dependent expression that can be used to reduce efficiently the overall delivery time. The approach is generalized in \cite{7037034} to decentralized distributed system using game theory tools.

Most of the previous works, optimizing different parameters in IDNC-based systems, consider an upper-layer view of the network and abstract its physical-layer conditions, e.g., fading, shadowing, etc., into simple erasure channel models. Moreover, all coded combinations from the base-station (BS) are assumed to be transmitted with the same physical-layer rate, thus occupying fixed durations of time. Such assumptions definitely simplified the modeling of IDNC scenarios at the expense of the impracticality of the channel model. Indeed, it is well-known that different users in a cellular network undergo different channel conditions and thus shall be served at various rates by the BS in order to receive the transmitted messages \cite{6874566}. Such variance in possible service rates indeed affects the decisions on both the signals to be combined, in each transmission, and the rate with which they are sent \cite{6662474}. Such rate adaptation affects not only the capability of different users to successfully receive each transmission, but it also determines the time duration needed to deliver the message.

The target of this paper is to study the completion time of delivering a group of messages to the network's users using rate-aware IDNC (RA-IDNC). This paper extends the previous upper-layer studies on IDNC in \cite{7037599,6809217,6882208}, by investigating the joint optimization of message combinations and employed rates in each of the transmissions, so as to minimize the overall completion time. Recently, rate adaptation in other network coding contexts is considered, e.g., \cite{6874566,6662474,6692154,6578154,5730588}. Nonetheless, this work is the first to study IDNC with rate awareness.

The paper's main contributions are as follows. The rate-aware IDNC completion time problem is first formulated and shown to be intractable as its upper-layer counterpart \cite{7037599,6809217,6882208}. Given this intractability, the paper considers a more manageable and online approximation of the completion time problem, called herein the anticipated completion time reduction problem. The paper formulates the problem as a maximum weight clique problem over a newly designed RA-IDNC graph. This new graph model incorporates both the IDNC possible coded combinations and the different transmission rates so as the communication is useful to the targeted subset of users. Parsing the considered problem as a maximum weight clique problem in the constructed graph provides efficient solvers for it, e.g., \cite{13265492,6607889}. Since the maximum weight clique in the designed graph may not be unique, the paper further suggests a multi-layer version of the proposed solution to improve the obtained results from the employed completion time approximation. The proposed solution is then tested and compared to classical upper-layer IDNC and other uncoded strategies using extensive simulations.

The rest of the paper is organized as follows: \sref{sec:sys} presents the system model and relevant definitions. Section \sref{form} illustrates the intractability of the rate-aware completion time problem and \sref{sec:com} presents a more tractable approximation of this problem. The proposed cross-layer algorithm using the designed RA-IDNC graph is, then, developed in \sref{sec:pro}. Finally, \sref{sec:sim} presents the simulation results before concluding in \sref{sec:con}.

\section{System Model and Definitions} \label{sec:sys}

\subsection{Network and Physical Layer Model}

Consider the downlink of a wireless radio network with a single BS. The BS is required to deliver a set $\mathcal{F}$ of messages to a set $\mathcal{U}$ of users. The messages in $\mathcal{F}$ are all assumed to be of an equal size of $N$ bits and could represent files, frames from a video stream, etc. The paper assumes that the time is indexed, such that the time index $t \in \mathds{N}^+$ corresponds to the starting time of the $t$-th transmission of a message (or a coded message combination) from the BS. In the rest of the paper, the notation $|\mathcal{X}|$ denotes the cardinality of the set $\mathcal{X}$.

Let $h_u(t)$ be the complex channel gain from the BS to the $u$-th user at the $t$-th transmission. It is assumed that $h_u(t),\ \forall~u \in \mathcal{U}$ remains constant during the $t$-th transmission. Let $P$ be the transmit power of the BS assumed to be fixed. The achievable capacity of the $u$-th user in the network during the $t$-th transmission can be written as:
\begin{align}
R_u(t) = \log_2 \left( 1 + \cfrac{P |h_u(t)|^2}{\Gamma(\sigma^2)}\right),
\end{align}
where $\Gamma$ denotes the capacity gap, and $\sigma^2$ denotes the Gaussian noise variance. Let $\mathcal{R} (t)= \{R_1(t),\ \dots,\ R_{|\mathcal{U}|}(t)\}$ be the set of the achievable capacities of all users during the $t$-th transmission.

The \emph{absolute time}\footnote{The term \emph{absolute time} is used to both designate the physical time and differentiate it from the indexed time denoting the beginning of the transmissions.} required to transmit a message of size $N$ in the $t$-th transfer using any selected rate $R(t)$ is $N/R(t)$. This paper assumes that the BS can correctly adjust its modulation scheme to target any value of the rate $R(t)$. Further, without loss of generality, it is assumed that the value of the capacity gap $\Gamma$ is chosen appropriately such that it enables perfect modulation. This means that the $u$-th user will not be able to receive the $t$-th transmission if $R(t) > R_u(t)$.

Ideally, the $t$-th transmission results in a successful reception at the $u$-th user if and only if the rate of the transmission is lower than the capacity, i.e., $R(t) \leq R_u(t)$. However, due to channel estimation errors, the user's capacity may be misestimated. To mitigate the effect of the estimation errors, let $\epsilon_u(R,R_u)$ be the message erasure probability at the $u$-th user with a capacity $R_u$ for a transmission at a rate $R$. In other words, the transmission at a rate $R$ for the user $u$-th user with capacity $R_u$ is erased with a probability of $\epsilon_u(R,R_u)$. Note that the particular case
\begin{align}
\epsilon_u(R,R_u)=
\begin{cases}
0 \hspace{0.5cm} &\text{if } R \leq R_u \\
1 \hspace{0.5cm} &\text{otherswise.}
\end{cases}
\label{eqepsi}
\end{align}
translates the perfect channel estimation scenario.

\subsection{Data Model}

The paper assumes that each user is interested in receiving all the messages in $\mathcal{F}$. At any given instant of time, these messages can be decomposed into the two sets from the $u$-th user perspective:
\begin{itemize}
\item The \emph{Has} set $\mathcal{H}_u$ containing the messages successfully received by the $u$-th user.
\item The \emph{Wants} set $\mathcal{W}_u = \mathcal{F} \setminus \mathcal{H}_u$ containing the messages missing at the $u$-th.
\end{itemize}

Initially, no user has any of these messages, i.e., $\mathcal{H}_u = \varnothing$ and $\mathcal{W}_u = \mathcal{F},\ \forall~u\in\mathcal{U}$. Therefore, the BS starts by sending each of these messages uncoded, with rates chosen according to the employed scheme\footnote{The rest of the paper explains only the considered RA-IDNC scheme. The rate selection of the other schemes can be found in the simulation results section.}. As explained above, the $t$-th message transmission is successfully received by the $u$-th user with a probability $1-\epsilon_u(R,R_u)$.

After one full round of transmitting messages uncoded, the BS exploits the diversity in the Has and Wants sets of the different users to send XOR-coded combinations of their missing messages. A user can decode a message $f$ from a coded message combination only if all the other messages in the combination are in its Has set. Indeed, such user can XOR these messages with the received coded message combination to extract $f$. From the instant decodability restriction of IDNC, a message combination with more than one unknown message for a given user is useless for that user, and thus it is discarded.

\begin{remark}
In the rest of the paper, the term ``transmission'' refers to both the process and the whole duration of sending any message by the BS.
\end{remark}

\subsection{Important Definitions}

This subsection gathers several definitions of the terms that are used throughout the paper.

\begin{definition}[Instantly Decodable Transmission]
A transmission $(\kappa(t),R(t))$ is instantly decodable for the $u$-th user if it is both communicated at a rate $R(t) \leq R_u(t)$ and the message combination $\kappa(t)$ contains only one message from $\mathcal{W}_u$.
\end{definition}

\fref{fig:SFM} illustrates an example of an IDNC transmission in a network composed of $3$ users and $3$ messages assuming the transmission rate is less than the capacities of all users. The packet combination $2 \oplus 3$ is:
\begin{itemize}
\item Non-instantly decodable for user $1$ as all the messages are in his Has set. The message does not bring new information.
\item Non-instantly decodable for user $2$ as it contains $2$ messages from his Wants set. User $2$ discards the message upon successful reception.
\item Instantly decodable for user $3$ as it contains only one message from $\mathcal{W}_3$. Indeed, users $3$ can XOR the combination $2 \oplus 3$ with message $3$ to retrieve message $2$.
\end{itemize}

\begin{figure}[t]
\centering
\includegraphics[width=1\linewidth]{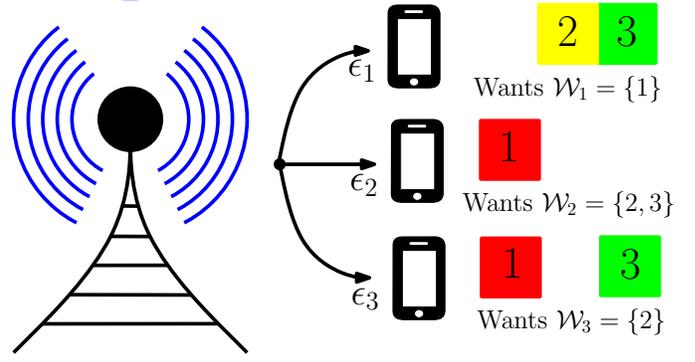}\\
\caption{A network composed of $3$ users and $3$ messages. It can be easily seen that the message combination $1 \oplus 2$ provides a new packet for all users. However, due to the rate asymmetry between users, it may not be optimal as shown in \fref{fig:graph}}\label{fig:SFM}
\end{figure}

\begin{definition}[Individual Completion Time]
The individual completion time $\mathcal{C}_u$ of the $u$-th user is the absolute time required until the user obtains all its the packets in $\mathcal{F}$.
\end{definition}

\begin{definition}[Overall Completion Time]
The overall completion time $\mathcal{C}$ is the absolute time required until all users obtain all the messages in $\mathcal{F}$, i.e., $\mathcal{C} = \max_{u\in\mathcal{U}} \{\mathcal{C}_u\}$.
\end{definition}

\begin{definition}[Transmission Schedule]
A transmission schedule $\mathcal{S} = \{(\kappa(t),R(t))\}$, $\forall~t\in\{1,\ \dots,\ |\mathcal{S}|\}$ as the set of chosen message combinations and rates at every time index $t$ until the overall completion time is reached. Further, define $\mathbb{S}$ as the set of all possible transmission schedules.
\end{definition}

\begin{definition}[Time Delay]
In the $t$-th transmission, the $u$-th user, with non-empty Wants set, experiences $N/R(t)$ seconds of time delay increase if he successfully receives a transmission that is non-instantly decodable for that user. Further, define the accumulated time delay $\mathcal{T}_u(\mathcal{S})$ experienced by the $u$-th user as the sum of all such time delay increases experienced by the $u$-th user throughout the schedule $\mathcal{S}$ until its individual completion time.
\end{definition}

Assuming error-free transmissions at a rate of $1$ bit/sec that is less than the capacities of all users in the example illustrated in \fref{fig:SFM}, the schedule of messages $\mathcal{S} =\{2 \oplus 3, 1, 1 \oplus 3,2\}$ yields the following:
\begin{itemize}
\item Individual completion time: $\mathcal{C}_1=N$, $\mathcal{C}_2=4N$, and $\mathcal{C}_3=2N$ seconds.
\item Overall completion time: $\mathcal{C}=4N$ seconds.
\item Time delay: $\mathcal{T}_1=N$, $\mathcal{T}_2=2N$, and $\mathcal{T}_3=0$ seconds.
\end{itemize}

On the other hand, for a fixed rate $R=1$, the optimal schedule $\mathcal{S}^* =\{1 \oplus 2,3\}$ provides the least individual and overall completion time of $\mathcal{C}=2N$ and zeros delay for all users.

\section{Problem Formulation} \label{form}

The aim of this paper is to investigate the problem of finding the optimal schedule $\mathcal{S}^*$, so as to minimize the overall completion time $\mathcal{C}(\mathcal{S}^*)$. From the above definitions, the problem can be expressed as shown in the following theorem.
\begin{theorem}
The minimum overall completion time problem in rate aware IDNC reduces to finding the optimal schedule $\mathcal{S}^*$, such that:
\begin{align}
\mathcal{S}^* \approx \arg\min_{\mathcal{S} \in \mathbb{S}}\left\{ \max_{u\in\mathcal{U}}\left\{ \left( \cfrac{N|\mathcal{F}|}{\tilde{R}_u(\mathcal{S})} + \mathcal{T}_u(S) \right) \left( 1 - \overline{\epsilon}_u \right)^{-1}\right\}\right\}\;,
\label{eq25}
\end{align}
where $\tilde{R}_u(\mathcal{S})$ is the harmonic mean of the rates of the successfully received transmissions that are instantly decodable for the $u$-th user in schedule $\mathcal{S}$ and where $\overline{\epsilon}_u = \mathds{E}_{R(t)\leq R_u(t)}[\epsilon_u(R(t),R_u(t))]$ is the average erasure probability for the transmission with a rate $R(t)\leq R_u(t)$.
\label{th1}
\end{theorem}
\begin{proof}
To show this theorem, the individual completion time is first expressed as a sum of the instantly decodable transmissions time, the time delay, and the erased transmissions. Afterward, the time of immediately decodable transmissions is approximated using the mean transmission time that can be written as a function of the harmonic mean of the rates of the transmissions that are instantly decodable for the user. Finally, to conclude the proof, the time of the erased transmission is approximated using the individual completion time of the user. The complete proof can be found in \appref{app1}.
\end{proof}

\begin{remark}
The rest of the paper uses the approximation proposed in \thref{th1} as equality for it holds for a large number of transmissions and/or a nearly constant erasure probability.
\end{remark}

The following corollary provides the expression of the completion time for a perfect channel estimation.
\begin{corollary}
The optimal schedule $\mathcal{S}^*$ that minimizes the overall completion in rate aware IDNC with perfect channel estimation can be, precisely, written as follows:
\begin{align}
\mathcal{S}^* = \arg\min_{\mathcal{S} \in \mathbb{S}}\left\{ \max_{u\in\mathcal{U}} \left( \cfrac{N|\mathcal{F}|}{\tilde{R}_u(\mathcal{S})} + \mathcal{T}_u(S) \right) \right\}.
\label{eqcor2}
\end{align}
\label{cor2}
\end{corollary}

\begin{proof}
As shown in the previous section, the perfect channel estimation scenario can be obtained by considering the erasure probability of the $u$-th user as follows:
\begin{align}
\epsilon_u(R,R_u)=
\begin{cases}
0 \hspace{0.5cm} &\text{if } R \leq R_u \\
1 \hspace{0.5cm} &\text{otherswise.}
\end{cases}
\end{align}
It can easily be seen that the average erasure likelihood of the $u$-th user for transmissions verifying $R(t)\leq R_u(t)$ is zero, i.e., :
\begin{align}
{\epsilon}_u = \mathds{E}_{R(t)\leq R_u(t)}[\epsilon_u(R(t),R_u(t))] = 0
\end{align}
Therefore, the optimal schedule that reduces the completion time can be expressed as follows:
\begin{align}
\mathcal{S}^* \approx \arg\min_{\mathcal{S} \in \mathbb{S}}\left\{ \max_{u\in\mathcal{U}} \left( \cfrac{N|\mathcal{F}|}{\tilde{R}_u(\mathcal{S})} + \mathcal{T}_u(S) \right) \right\}.
\end{align}
Finally, applying the weak law of large number to a constant quantity yields an exact expression instead of an approximation that explain the equality in \eref{eqcor2}
\end{proof}

Finding the global optimal solution to the optimization problem \eref{eq25} is clearly intractable due to the dynamic nature of the channel realizations and the interdependence of optimal schedule on them. Such relationship makes the decision depends on future inputs, i.e., it makes the problem anti-causal. Further, the determination of the optimal schedule, prior to its start, is shown to be very difficult, even without rate adaptation \cite{7037599,6809217,6882208}. Given the complexity of finding the optimal schedule, this paper reformulates the problem by a more tractable approximation of the completion times called herein the \emph{anticipated completion times}.

\section{Anticipated Completion Time Formulation} \label{sec:com}

\subsection{Definition of Anticipated Completion Times}

Let the anticipated individual completion time (denoted by $\mathcal{C}_u(t)$) of the $u$-th user after the $t$-th transmission be its completion time if that user does not experience any additional increases in its time delay, i.e., it does not receive any other non-instantly decodable transmissions, starting from the $t$-th transmission. From the expression of $\mathcal{C}_u(\mathcal{S})$ in \eref{eq25}, it is easy to infer that $\mathcal{C}_u(t)$ can be approximated by:
\begin{align}
\mathcal{C}_u(t) = \left(\cfrac{N|\mathcal{F}|}{\tilde{R}_u(t)} + \mathcal{T}_u(t)\right) \left( 1 - \tilde{\epsilon}_u(t) \right)^{-1},
\label{ct}
\end{align}
where $\tilde{R}_u(t)$ is the harmonic mean of the rates of all transmissions $t^\prime \leq t$ that are successfully received and instantly-decodable for the $u$-th user, $\mathcal{T}_u(t)$ is the accumulated time delay experienced by the user until the $t$-th transmission, and $\tilde{\epsilon}_u$ is the average erasure probability for $R(k) \leq R_u(k), 1 \leq k \leq t$. In fact, the above expression of $\mathcal{C}_u(t)$ of the $u$-th user would have been exact if $\tilde{R}_u(t) = \tilde{R}_u(\mathcal{S})$ and $\tilde{\epsilon}_u(t) = \overline{\epsilon}_u$.

\subsection{Problem Approximation}

This subsection approximates the optimization problem \eref{eq25} by a more tractable version using the definition of anticipated completion times in \eref{ct}. Let $u^*$ be the user with the maximum anticipated completion time at transmission $t-1$, i.e., $u^* = \arg \max_{ u \in \mathcal{U}}\{\mathcal{C}_u(t-1)\}$. Let $\mathcal{K}_{R(t)}$ be the set of decisive users defined as follows:
\begin{align}
\mathcal{K}_{R(t)} = \left\{ u \in \mathcal{U} ~\middle|~ \ \mathcal{C}_u(t-1) + \frac{N}{R(t)} \geq \mathcal{C}_{u^*}(t-1) \right\}.
\label{kk}
\end{align}

Clearly, $\mathcal{K}_{R(t)}$ is the set of users that can increase the maximum anticipated overall completion time, from its maximum value after the $(t-1)$-th transmission, i.e., $\mathcal{C}_{u^*}(t-1)$. Indeed users in $\mathcal{K}_{R(t)}$ increase the maximum anticipated completion time if the $t$-th transmission using rate $R(t)$ is non-instantly decodable for them. The derivation of this set can be found in \appref{app2}.

Given the above definition, the completion time minimization problem \eref{eq25} is approximated by the online optimization problem of finding the message combination and the transfer rate, at each transmission $t$, so as to minimize the expected increase of the overall anticipated completion time from the $(t-1)$-th transmission to the $t$-th one as shown in the following lemma.
\begin{lemma}
The completion time reduction problem \eref{eq25} can be approximated, at each transmission $t$, by the following joint optimization problem over the message combination $\kappa(t)$ and the transfer rate $R(t)$:
\begin{align}
&(\kappa^*(t), R^*(t)) \label{opti} \\
& \qquad = \arg \max_{ \substack{\kappa(t) \in \mathcal{P}(\mathcal{F}) \\ R(t) \in \mathcal{R}(t)}} \ \sum_{u \in \mathcal{K}_{R(t)} \cap \tau_{\kappa(t)}} \log \left(\frac{R(t)}{\epsilon_u(R,R_u)N} \right), \nonumber
\end{align}
where the notation $\mathcal{P}(\mathcal{X})$ refers to the power set of the set $\mathcal{X}$ and $\tau_{\kappa(t)}$ denotes the set of users for those the transmission the message combination $\kappa(t)$ at a rate $R(t)$ is instantly decodable.
\label{l1}
\end{lemma}

\begin{proof}
To demonstrate the lemma, the set $\mathcal{K}_{R(t)}$ of users that can potentially increase the anticipated version of the completion time is first identified and its expression \eref{kk} established. The set of such users is used to approximate the problem as a problem of minimizing the increase in the anticipated completion time. Finally, the joint optimization problem over the message combination and the transfer rate is derived. The complete proof can be found in \appref{app2}.
\end{proof}

The following corollary introduces the joint optimization problem over the message combination and the transmission rate for a perfect channel estimation scenario.
\begin{corollary}
The completion time reduction problem with perfect channel estimation can be approximated by the following online optimization problem over the message combination $\kappa(t)$ and the transmission rate $R(t)$:
\begin{align}
&(\kappa^*(t), R^*(t)) = \arg \max_{ \substack{\kappa(t) \in \mathcal{P}(\mathcal{F}) \\ R(t) \in \mathcal{R}(t)}} \ \sum_{u \in \mathcal{K}_{R(t)} \cap \tau_{\kappa(t)}} \log \left(\frac{R(t)}{N} \right). \nonumber
\end{align}
\end{corollary}

\begin{proof}
To establish the expression of the expected completion time increase in a system with perfect channel estimation, it is sufficient to note that for any message combination $\kappa(t)$ and any user $u \in \tau_{\kappa(t)}$, the transmission is instantly decodable for that user. Therefore, from the definition of the instant decodability of a message, the rate of the transfer satisfy $R(t) \leq R_u(t)$. Therefore, from the expression of the erasure probability in perfect channel estimation scenario \eref{eqepsi}, $\epsilon_u(R,R_u) = 0, \forall \ u \in \tau_{\kappa(t)}$, i.e., constant independent from both $\kappa(t)$ and $R(t)$. Therefore, the joint optimization problem can be written as:
\begin{align}
&\max_{ \substack{\kappa(t) \in \mathcal{P}(\mathcal{F}) \\ R(t) \in \mathcal{R}(t)}} \ \sum_{u \in \mathcal{K}_{R(t)} \cap \tau_{\kappa(t)}} \log \left(\frac{R(t)}{N} \right) - \log(\epsilon_u(R,R_u)) \nonumber \\
&= \max_{ \substack{\kappa(t) \in \mathcal{P}(\mathcal{F}) \\ R(t) \in \mathcal{R}(t)}} \ \sum_{u \in \mathcal{K}_{R(t)} \cap \tau_{\kappa(t)}} \log \left(\frac{R(t)}{N} \right)
\end{align}
\end{proof}

\section{Proposed Solution} \label{sec:pro}

This section proposes solutions to the approximated completion time reduction problem \eref{opti} by modeling it in the form of a graph that will be referred to as the rate-aware IDNC (RA-IDNC) graph. In the next two subsections, the time index notation is dropped, as it is known that the parameters are expressed for the $t$-th transmission.

\subsection{RA-IDNC Graph}

In the network-layer configuration of \cite{5683677}, the authors introduced the IDNC graph as a tool to determine all possible message combinations and identify the users that can instantly decode each of these combinations. This subsection extends the graph formulation to the rate-aware instantly decodable network coding scenario of interest in this paper. The presented graph formulation allows the identification of the messages combinations, transmission rate, and the set of users for those the transmission is instantly decodable. This graph will be termed as the RA-IDNC graph, denoted by $\mathcal{G}(\mathcal{V},\mathcal{E})$ ($\mathcal{V}$ and $\mathcal{E}$ being the set of vertices and edges of the graph, respectively), and constructed as follows.

To generate the set of vertices, first introduce the set of achievable rates $\mathcal{R}_u = \{ R \in \mathcal{R} \ | \ R \leq R_u\}$ for each user $u$. In other words, for each user, the set of achievable rate is the highest rates observed by other users whose rate region is a subset of that user. A vertex $v_{ufr}$ is created for each feasible association of user $u \in \mathcal{U}$, a wanted message $f \in \mathcal{W}_u$ and an achievable rate from $\mathcal{R}_u$. The set of edges $\mathcal{E}$ is generated by connecting two vertices $v_{u f r }$ and $v_{u^\prime f^\prime r^\prime }$ if they satisfy both following conditions:
\begin{itemize}
\item C1 : $r = r^\prime$.
\item C2 : $f = f^\prime$ OR $(f,f^\prime) \in \mathcal{H}_{u^\prime} \times \mathcal{H}_{u}$.
\end{itemize}

The connectivity condition C1 ensures that the rate of the transmission is the same for all adjacent vertices in the RA-IDNC graph. The connectivity condition C2 represents the instant decodability condition of IDNC from the network-layer perspective. The late can be decomposed in two sub-conditions namely $f = f^\prime$ that translates the fact that the same message is requested by distinct users and $(f,f^\prime) \in \mathcal{H}_{u^\prime} \times \mathcal{H}_{u}$ that states the instant decodability of the message mix $f \oplus f^\prime$ for both users $u$ and $u^\prime$.

Given the above construction rules, it can be readily inferred that each clique\footnote{A clique in an undirected graph is a set of vertices in which each two vertices are adjacent.} in the RA-IDNC graph represents a transmission both having:
\begin{itemize}
\item A message combination that is decodable to all the users designated by the clique's vertices.
\item One same rate that is smaller than the capacities of all the users identified by the clique's vertices.
\end{itemize}
Therefore, from Definition 1, the transmission represented by each clique in the RA-IDNC graph is instantly decodable for all the users designated by the clique's vertices. The following theorem characterizes the optimal solution to the approximated completion time reduction problem in \eref{opti} using the RA-IDNC graph:
\begin{theorem}
The optimal solution to the problem in \eref{opti} is the transmission that is represented by the maximum weight clique\footnote{In a weighted graph, the weight of a clique is defined as the sum of the individual weights of vertices belonging to the clique. The maximum weight clique problem is the one of finding the clique(s) with the maximum weight.} among all the maximal cliques\footnote{A maximal clique is a clique that is not a subset of any larger clique.} in the RA-IDNC graph, where the weight $w(v_{ufr})$ of each vertex $v_{ufr}$ is set to:
\begin{align}
w(v_{ufr}) =
\begin{cases}
\log \left( \frac{r}{\epsilon_u(r,R_u)N} \right) \hspace{0.5cm} &\text{if } u \in \mathcal{K}_{r} \\
0 \hspace{0.5cm} &\text{otherwise. }
\end{cases}
\label{weight}
\end{align}
\label{th2}
\end{theorem}
\begin{proof}
To prove this theorem, we first show that there exists a one-to-one mapping between the set of feasible message combination, transmission rate, and targeted users and the set of maximal cliques in the RA-IDNC graph. To conclude the proof, it is sufficient to show that the weight of the cliques is equivalent to the objective function of the optimization problem \eref{opti} under investigation. The complete proof can be found in \appref{app3}.
\end{proof}

The maximum weight clique problems are a well-known NP-hard problems. However, they can be solved more efficiently \cite{13265492,6607889} than the $ \mathcal{O}(|\mathcal{V}|^2\cdot 2^{|\mathcal{V}|})$ naive exhaustive search solving methods. Further, various approximate solutions \cite{9874286} yield, in general, acceptable outcomes. \fref{fig:graph} provides an example of the RA-IDNC graph for a network composed of $3$ users and $3$ messages. Note that in classical IDNC, the rate of the transmission should be the minimum rate of the targeted users by the transmission.

\begin{figure}[t]
\centering
\includegraphics[width=1\linewidth]{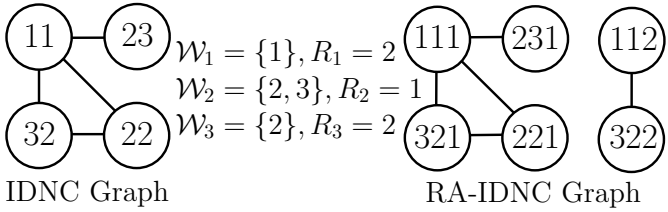}\\
\caption{Comparision between the classical IDNC graph and the RA-IDNC graph. In classical IDNC, the clique $\{11,22,32\}$ is selected because it serves $3$ users. However, the transmission provides $3$ bits/s (rate times number of trageted users) to the network. In RA-IDNC, the clique $\{112,322\}$ is selected bacause it provides $4$ bits/s.}\label{fig:graph}
\end{figure}

\subsection{Multi-Layer Solution}

\begin{algorithm}[t]
\begin{algorithmic}
\REQUIRE $\mathcal{U},\mathcal{F}, N , \mathcal{H}_u, \mathcal{W}_u, \mathcal{R}_u, \mathcal{C}_u(t-1), \ \forall \ u \in \mathcal{U}$.
\STATE Initialize $\mathbf{M} =\varnothing$.
\STATE Compute $\mathcal{K}^1 = \cup_{R \in \mathcal{R}} \{\mathcal{K}^1_{R}\}$.
\FOR{$k=1$ \TO $\cdots$}
\STATE Construct $\mathcal{G}^k(\mathbf{M})$.
\STATE Solve the maximum weight clique problem to yield $\mathbf{M}^k$.
\STATE Sets $\mathbf{M} = \mathbf{M} \cup \mathbf{M}^k$.
\STATE Compute $\mathcal{K}^{k+1}_R$, $R$ rate of $v \in \mathbf{M}$.
\ENDFOR
\STATE Output $\mathbf{M}$.
\end{algorithmic}
\caption{Multi-Layer RA-IDNC Graph}
\label{algo1}
\end{algorithm}

As shown in the previous subsection, the optimal solution to the optimization problem \eref{opti} can be obtained by identifying the maximum weight clique(s) among all maximal cliques in the RA-IDNC graph. This subsection improves upon the proposed solution. The improvement is fundamentally based on the principle that the optimal solution of \eref{opti} is potentially (in most cases) not unique. Therefore, this subsection relies on an efficient choice of one of the points that achieve the optimum of \eref{opti}. Such point is chosen by focusing on message combinations serving users that are the most likely to be decisive users in the next transmissions.

First, introduce the set of $k$-th critical users $\mathcal{K}_{R}^k$ as the set of users that can potentially increase the anticipated overall completion time if and only if they experience $k \geq 1$ consecutive time delay increases, i.e., they receive $k$ subsequent non-instantly decodable transmissions at a rate $R$, starting from the $t$-th transmission. The mathematical definition of this set is:
\begin{align}
\mathcal{K}_{R}^k = \left\{ u \notin \mathcal{K}_{R}^{k-1} ~\middle|~ \mathcal{C}_u(t-1) + \frac{kN}{R} \geq \mathcal{C}_{u^*}(t-1) \right\} , \nonumber
\end{align}
with $\mathcal{K}_{R}^1=\mathcal{K}_{R}$ defined in \eref{kk}.

Let $\mathcal{G}^1$ be the RA-IDNC graph constructed using only the users in the set $\mathcal{K}^1 = \cup_{R \in \mathcal{R}} \{\mathcal{K}^1_{R}\}$. In other words, a vertex $v_{ufr}$ is created for each feasible association of user $u \in \mathcal{K}^1$, a wanted message $f \in \mathcal{W}_u$ and an achievable rate in $\mathcal{R}_u$. The connectivity condition and the weights are the same as for the RA-IDNC graph introduced in the previous subsection. The following corollary characterizes the maximum weight clique among the maximal cliques in $\mathcal{G}^1$:
\begin{corollary}
The maximum weight clique among the maximal cliques in the $\mathcal{G}^1$ RA-IDNC graph yields the same weight as the one(s) in the $\mathcal{G}(\mathcal{V},\mathcal{E})$ RA-IDNC graph.
\label{cor1}
\end{corollary}

\begin{proof}
The proof of this corollary is straightforward. It is sufficient to note that $\mathcal{G}^1 \subset \mathcal{G}$ and that $\forall \ v_{fur} \in \mathcal{G} \setminus \mathcal{G}^1$, we have $w(v_{ufr}) = 0$. Therefore, the weight of the maximum weight clique $\mathbf{M}^1$ in $\mathcal{G}^1$ as function of the maximum weight clique $\mathbf{M}$ in $\mathcal{G}$ is:
\begin{align}
w(\mathbf{M}^1) = w(\mathbf{M}) - \sum_{v_{ufr} \in \mathbf{M} \cap (\mathcal{G} \setminus \mathcal{G}^1)} w(v) = w(\mathbf{M}).
\end{align}
\end{proof}

Let $\mathbf{M}^1$ be the maximum weight clique among the maximal cliques in $\mathcal{G}^1$. From \corref{cor1}, each set of vertices in $\mathcal{G} \setminus \mathcal{G}^1$ connected to the vertices in $\mathbf{M}^1$ form a global optimal solution to \eref{opti}. Thus, the chosen clique $\mathbf{M}$ is initialized to $\mathbf{M} = \mathbf{M}^1$. In order to prioritize users that can potentially increase the anticipated completion time after two transmissions, let $\mathcal{G}^2(\mathbf{M})$ be the RA-IDNC graph constructed using vertices in $\mathcal{K}^2_R$ that are connected to all the vertices in $\mathbf{M}$, where $R$ is the rate of users selected in $\mathbf{M}$. After two transmissions, such vertices can potentially become decisive users with equal weight equal to $\log ( R/N )$. Hence, the weight of all the vertices in $\mathcal{G}^2(\mathbf{M})$ is $\log (R/N )$. The maximum weight clique problem (which is equivalent to a maximum clique problem in this case) is solved in $\mathcal{G}^2(\mathbf{M}^1)$ to produce the clique $\mathbf{M}^2$. The clique $\mathbf{M}^2$ is merged with $\mathbf{M}$ to form the updated clique $\mathbf{M} = \mathbf{M}^1 \cup \mathbf{M}^2$. The process is, then, repeated by generating $\mathcal{G}^3(\mathbf{M})$ and solving the maximum (weight) clique problem in it, and so forth until no further vertices in $\mathcal{G}$ are adjacent to all the vertices in $\mathbf{M}$. The steps of the algorithm are summarized in \algref{algo1}.

\section{Simulation Results} \label{sec:sim}

\begin{table}
\centering
\caption{Simulation parameters}
\begin{tabular}{|c|c|}
\hline
Cellular Layout & Hexagonal \\
\hline
Cell Diameter & 500 meters \\
\hline
Channel Model & SUI-3 Terrain type B \\
\hline
Channel Estimation & Perfect \\
\hline
High Power & -42.60 dBm/Hz \\
\hline
Background Noise Power & -168.60 dBm/Hz \\
\hline
SINR Gap $\Gamma$ & 0dB\\
\hline
Bandwidth & 10 MHz \\
\hline
\end{tabular}
\label{t1}
\end{table}

This section shows the performance of the proposed cross-layer algorithm in the downlink of a radio access network. The cell layout is hexagonal with $500$m of diameter. The BS is placed in the center of the cell, and the users uniformly distributed inside of it. The number of users, messages and the message size change in the simulation so as to study multiple scenarios. The shadowing variance is fixed to $0$dB in all the simulation except for the last one in which we consider the effect of such parameter. \tref{t1} summarizes the additional simulations parameters. The proposed rate aware IDNC is compared to the following schemes:
\begin{itemize}
\item The uncoded broadcast in which the message are transmitted uncoded to all users using their minimum achievable capacity as the rate of the transmission.
\item The uncoded unicast in which the user with maximum achievable capacity and non-empty Wants set is the only targeted user.
\item The classical IDNC in which the message combination is chosen according to the policy proposed in \cite{7037599}. For the transmission to be instantly decodable for all the targeted users, the transfer rate is the minimum achievable capacity of these users.
\end{itemize}

\begin{figure}[t]
\centering
\includegraphics[width=1\linewidth]{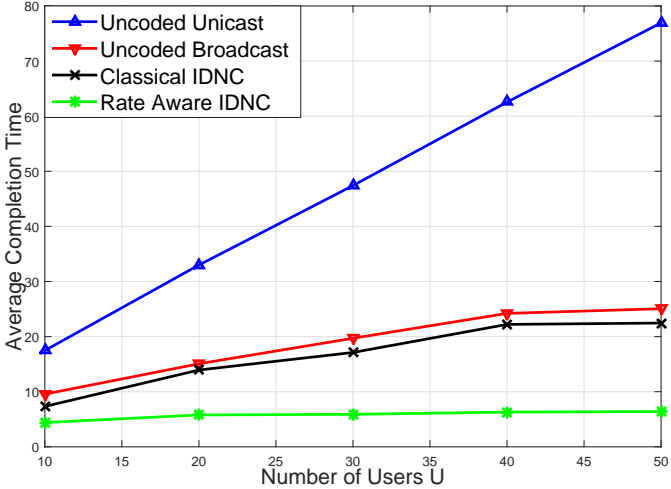}\\
\caption{Completion time in sec. versus the number of users $U$ for a network composed of $F=20$ messages, and a message's size $N=1$ Mb.}\label{fig:U}
\end{figure}

\begin{figure}[t]
\centering
\includegraphics[width=1\linewidth]{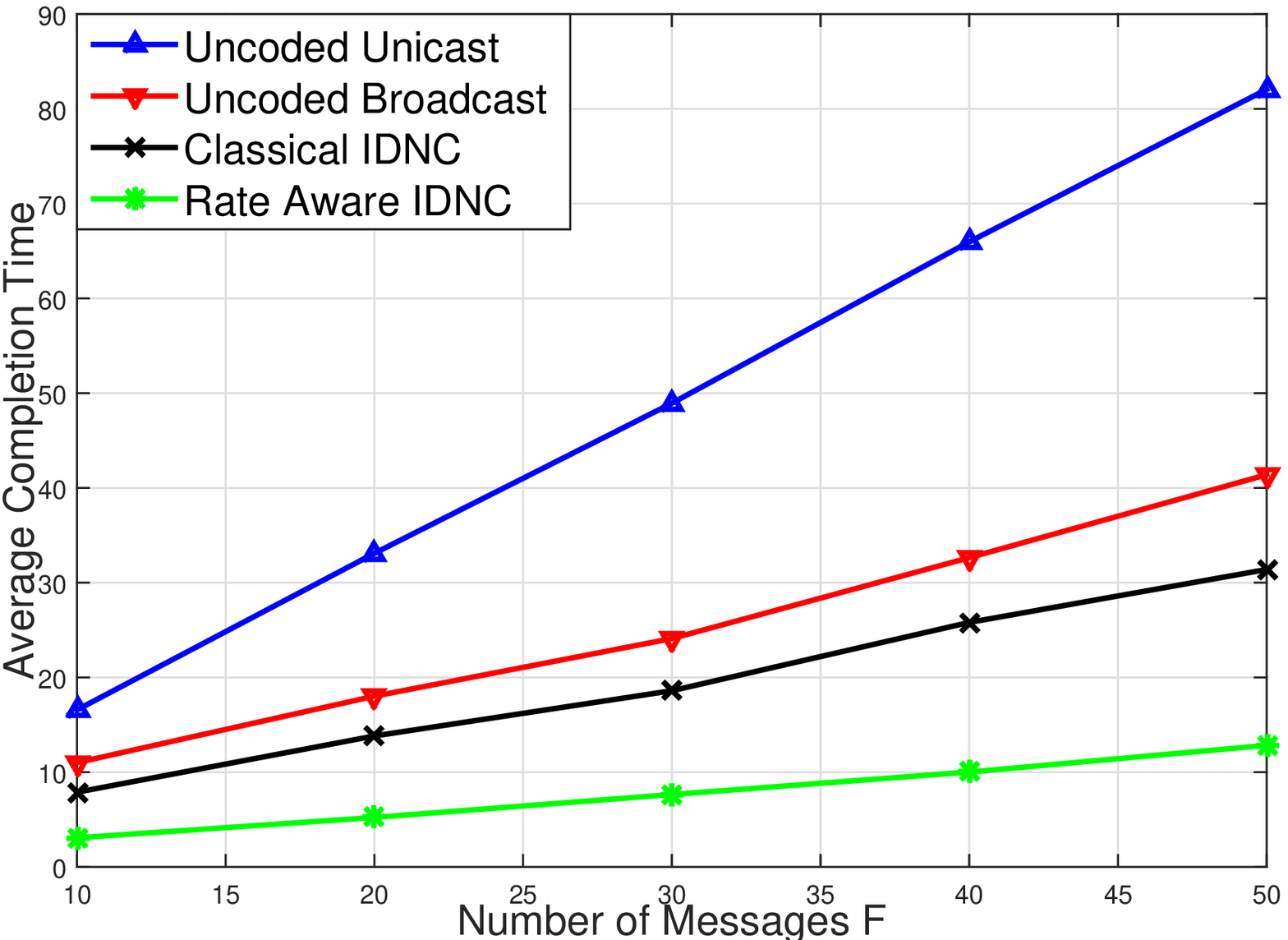}\\
\caption{Completion time in sec. versus the number of messages $F$ for a network composed of $U=20$ users and a message's size $N=1$ Mb.}\label{fig:F}
\end{figure}

\begin{figure}[t]
\centering
\includegraphics[width=1\linewidth]{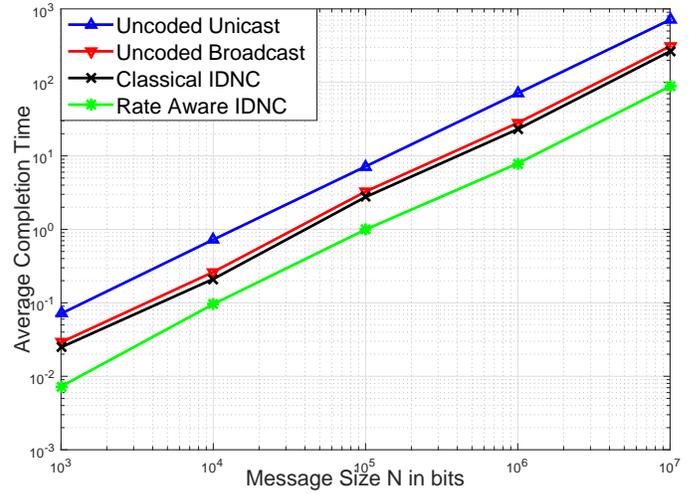}\\
\caption{Completion time in sec. versus the message's size $N$ for a network composed of $U=20$ users, and $F=30$ messages.}\label{fig:N}
\end{figure}

\begin{figure}[t]
\centering
\includegraphics[width=1\linewidth]{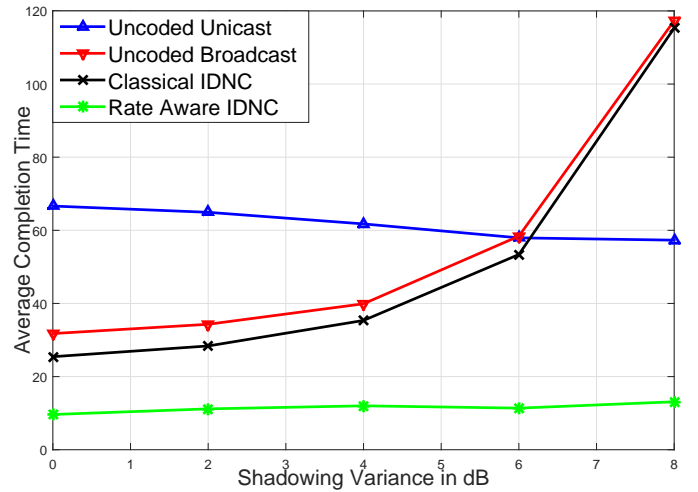}\\
\caption{Completion time in sec. versus the shadowing variance for a network composed of $U=20$ users, $F=40$ messages, and a message's size $N=1$ Mb.}\label{fig:V}
\end{figure}

First, \fref{fig:U} plots the completion time versus the number of users $U$ for a network composed of $F=20$ messages and a message's size $N=1$ Mb. We can clearly see that our proposed rate aware IDNC scheme largely outperforms the other schemes. The gap between our RA-IDNC and the classical IDNC increases as the number of users increases. This can be explained by the fact that as the number of users increases, the conventional IDNC scheme have more coding opportunities and thus the number of targeted users increases. This results in a lower transmission rate as it represents the minimum of an increasing set. The same thinking applies to the broadcast strategy.

\fref{fig:N} illustrates the completion time versus the number of messages $F$ for a network composed of $U=20$ users and a message's size $N=1$ Mb. Again, our proposed RA-IDNC outperforms the strategy without rate adaptation. We also note that the uncoded unicast scheme increases linearly with both the number of users and the number of messages. This can be explained by the nature of the scheme. As it targets only one user with one message at each transmission, it needs $UF$ transmissions before completion. Therefore, the performance of such systems is linear with the number of users (\fref{fig:U}) and the number of messages (\fref{fig:F}). The performance of the broadcast strategy is also linear with the number of messages since it requires $F$ transmissions to deliver all the messages.

\fref{fig:F} displays the completion time versus the message's size $N$ for a network composed of $U=20$ users, and $F=30$ messages. As hinted by \eref{eq25}, the completion time is the sum of a linear function of the message size and the time delay that is also linear with the message size. This result applies for any transmission scheme that explains that all the policies are linear with respect to such parameter.

Finally, to quantify the performance of the proposed scheme in high shadowing environment, \fref{fig:V} plots the completion time in sec. versus the shadowing variance in dB for a network composed of $U=20$ users, $F=40$ messages, and a message's size $N=1$ Mb. As the shadowing variance increases, the variability of the achievable capacity of different users increases. Since the unicast policy takes advantage of the maximum possible capacity, its performance improves as the shadowing variance increases. However, the classical and broadcast scheme are penalized with the minimum achievable capacity that decreases as the shadowing variance increases which explains the degradation in performance. As we can see from \fref{fig:V}, our proposed scheme strikes a balance between the number of targeted users and the transmission rate that explain its relatively constant performance against the shadowing variance variation.

\section{Conclusion} \label{sec:con}

This paper proposes solving the completion time reduction problem using instantly decodable network codes with rate adaptation. Unlike previous works, this paper suggests a cross-layer scheme that considers rate awareness in the transmission selection process so as to minimize further the completion time from its best-known value so far. The delivery time minimization problem is first shown to be intractable which motivates its approximation by the problem of reducing, at each transmission, the increase of an anticipated version of the completion time. The paper shows that the approximate problem is equivalent to a maximum weight clique problem in the newly constructed RA-IDNC. To further improve the performance of the proposed scheme, the paper introduces a multi-layer solution to refine the approximation. Simulation results suggest that the cross-layer design largely outperforms the uncoded transmissions strategies and the classical IDNC scheme.

\appendices

\numberwithin{equation}{section}

\section{Proof of \thref{th1}}\label{app1}

In this section, the optimal schedule is expressed as an optimization problem involving the time delay of users. First note that:
\begin{align}
\mathcal{S}^* &= \arg \min_{ \mathcal{S} \in \mathbb{S}} \left\{\mathcal{C}(\mathcal{S})\right\} = \arg\min_{\mathcal{S} \in \mathbb{S}}\left\{ \max_{u\in\mathcal{M}}\left\{ \mathcal{C}_u(\mathcal{S}) \right\}\right\}\;.
\end{align}
Therefore, to show this theorem it sufficient to show that the individual completion time can be approximated using the following expression:
\begin{align}
\mathcal{C}_u(\mathcal{S}) \approx \left( \cfrac{N|\mathcal{F}|}{\tilde{R}_u(\mathcal{S})} + \mathcal{T}_u(S) \right) \left( 1 - \overline{\epsilon}_u \right)^{-1}
\end{align}
In order to establish such approximation, this section first demonstrate that the individual time may be written as follows:
\begin{align}
\mathcal{C}_u(\mathcal{S}) = \cfrac{N|\mathcal{F}|}{\tilde{R}_u(\mathcal{S})} + \mathcal{T}_u(S) + \mathcal{E}_u(S),
\end{align}
where $\mathcal{E}_u(S)$ is the cumulative time of the transmissions that are erased at the $u$-th user transmitted at a rate $R \leq R_u$.

For a schedule $\mathcal{S}$ of message combination $\kappa(t)$ and transmit rates $R(t), 1 \leq t \leq |\mathcal{S}|$, the completion time $\mathcal{C}_u$ of user $u$ can be expressed as:
\begin{align}
\mathcal{C}_u(\mathcal{S}) = \sum_{t=1}^{n_u(\mathcal{S})} \cfrac{N(\kappa(t))}{R(t)},
\label{ctu}
\end{align}
where $N(\kappa(t))$ is the size of the message combination $\kappa$ and $n_u(\mathcal{S})$ is the first transmission index verifying $\mathcal{W}=\varnothing$. Since the messages are mixed using binary XOR, then the size of any message mixture is $N$. For instance $N(\kappa(t))=N$.

Let the side information of the base-station be stored in an $U \times F$ matrix $\mathbf{S}=[S_{uf}]$ as follows:
\begin{align}
S_{uf} =
\begin{cases}
0 \hspace{0.5cm} &\text{if }f \in \mathcal{H}_u \\
1 \hspace{0.5cm} &\text{if }f \in \mathcal{W}_u
\end{cases}, \ \forall \ (u,f) \in \mathcal{U} \times \mathcal{F},
\end{align}
where the notation $\mathcal{X} \times \mathcal{Y}$ refers to the Cartesian product of the two sets $\mathcal{X}$ and $\mathcal{Y}$. Let $S_u$ be the side information of user $u \in \mathcal{U}$, i.e., $S_u$ is the $u$th row of matrix $\mathbf{S}$.

The expression \eref{ctu} of the individual completion time of user $u$ given the schedule $\mathcal{S}$ of message combinations and the transmission rate can be formulated as follows:
\begin{align}
\mathcal{C}_u(\mathcal{S}) = \sum_{t \in \alpha_u(\mathcal{S}) } \cfrac{N(\kappa(t))}{R(t)} + \sum_{t \in \beta_u(\mathcal{S}) } \cfrac{N(\kappa(t))}{R(t)} + \sum_{t \in \gamma_u(\mathcal{S}) } \cfrac{N}{R(t)},
\label{eql11}
\end{align}
where the sets $\alpha_u(\mathcal{S})$, $\beta_u(\mathcal{S})$, and $\gamma_u(\mathcal{S})$ are the sets of the transmission index that are, respectively, received, instantly decodable, and non-instantly decodable for user $u$ and those erased at that user. The definition of these sets is given by:
\begin{align}
\alpha_u(\mathcal{S}) &= \{ t \leq n_u(\mathcal{S}) \ | \ S_u(t)^T k(t) = 1 \text{ and } R(t) \leq R_u(t)\} \nonumber \\
\beta_u(\mathcal{S}) &= \{ t \leq n_u(\mathcal{S}) \ | \ S_u(t)^T k(t) \neq 1 \text{ or } R(t) > R_u(t)\} \nonumber \\
\gamma_u(\mathcal{S})&= \{ t \leq n_u(\mathcal{S}) \ | \ \mathcal{X}_u(t) =1 \text{ and } R(t) \leq R_u(t) \},
\end{align}
where the notation $X^T$ refers to the transpose of vector $X$ and $\mathcal{X}_u(t)=1$ is the erasure variable that indicates that a message is erased at time $t$ for $R(t) \leq R_u(t)$.

To prove this theorem, we first show that, for any schedule $\mathcal{S}$, the cardinality of $\alpha_u(\mathcal{S})$ is $|\mathcal{F}|$, that $\sum_{t \in \beta_u(\mathcal{S}) } \cfrac{N}{R(t)}=\mathcal{T}_u(\mathcal{S})$, and that $\sum_{t \in \gamma_u(\mathcal{S}) } \cfrac{N}{R(t)} = \mathcal{E}_u(\mathcal{S})$. To conclude the proof, we substitute $\{R(t)\}_{t \in \alpha_u(\mathcal{S})}$ by its harmonic mean $\tilde{R}_u(\mathcal{S})$ and we approximate $\mathcal{E}_u(\mathcal{S})$.

Let $t^0_{\alpha}$ be the smallest time index in the set $\alpha_u(\mathcal{S})$. For $t < t^0_{\alpha_u}$, it is clear that $t \notin \alpha_u(\mathcal{S})$ and hence the transmissions are non-instantly decodable for user $u$. Therefore, for $t < t^0_{\alpha_u}$, the Wants set of user $u$ is unchanged, i.e., $\mathcal{W}_u(t) = \mathcal{F}$. Now let $t_{\alpha}$ and $t^\prime_{\alpha}$ be any two consecutive time indices in the set $\alpha_u(\mathcal{S})$. Using the same steps as above, it is clear that $\mathcal{W}_u(t) = \mathcal{W}_u(t_{\alpha})$ for any $t_{\alpha} \leq t < t^\prime_{\alpha}$. The transmission at time $t^\prime_{\alpha}$, being an instantly decodable one, provides a new message for user $u$. Therefore, it reduces its Wants set by one unit, i.e., $\mathcal{W}_u(t^\prime_{\alpha}) = \mathcal{W}_u(t_{\alpha})-1$. Finally, from the definition of $n_u(\mathcal{S})$, we have $\mathcal{W}_u = \varnothing$. Since $n_u(\mathcal{S})$ is the smallest index making the Wants set empty, then it is easy to establish that $n_u(\mathcal{S}) \in \alpha_u(\mathcal{S})$ which concludes that
\begin{align}
|\alpha_u(\mathcal{S})| = |\mathcal{F}|.
\label{eql12}
\end{align}

Let $t < n_u(\mathcal{S})$ be a time index before the completion time of user $u$. From the analysis above, three scenarios can be distinguished:
\begin{itemize}
\item $t \in \alpha_u(\mathcal{S})$: The message is instantly decodable, and hence the user do not experience any delay.
\item $t \in \beta_u(\mathcal{S})$: The message is non-instantly decodable and $\mathcal{W}_u \neq \varnothing$ (since $t < n_u(\mathcal{S})$). Such transmission increases the time delay by $N/R(t)$.
\item $t \in \gamma_u(\mathcal{S})$: The message is erased, and hence the user do not experience any delay.
\end{itemize}
Given that $t \geq n_u(\mathcal{S})$ the user have an empty Wants set and by definition of the time delay, such user will not experience any additional delay. Therefore, we conclude that:
\begin{align}
\sum_{t \in \beta_u(\mathcal{S}) } \cfrac{N}{R(t)}=\mathcal{T}_u(\mathcal{S}).
\label{eql13}
\end{align}
By definition of the set of erased transmissions $\mathcal{E}_u(\mathcal{S})$ and the set of transmission index $\gamma_u(\mathcal{S})$, it can easily be concluded that
\begin{align}
\sum_{t \in \gamma_u(\mathcal{S}) } \cfrac{N}{R(t)} = \mathcal{E}_u(\mathcal{S}).
\label{eql15}
\end{align}

Let $\tilde{R}_u(\mathcal{S})$ be the harmonic mean of $\{R(t)\}_{t \in \alpha_u(\mathcal{S})}$. The relationship linking the two quantities can be written as:
\begin{align}
\cfrac{|\alpha_u(\mathcal{S})|}{\tilde{R}_u(\mathcal{S})} = \sum_{t \in \alpha_u(\mathcal{S}) } \cfrac{1}{R(t)}.
\label{eql14}
\end{align}

Substituting the equalities \eref{eql12}, \eref{eql13}, \eref{eql15}, and \eref{eql14} in the expression \eref{eql11} of the individual completion time of user $u$ yields:
\begin{align}
\mathcal{C}_u(\mathcal{S}) = \cfrac{N|\mathcal{F}|}{\tilde{R}_u(\mathcal{S})} + \mathcal{T}_u(S) + \mathcal{E}_u(S).
\label{eql16}
\end{align}

Now we approximate $\mathcal{E}_u(S)$ by the quantity $\overline{\epsilon}_u\mathcal{C}_u(\mathcal{S})$. The quantity in \eref{eql15} can be rewritten as follows:
\begin{align}
\mathcal{E}_u(\mathcal{S}) = \sum_{t \in \gamma_u(\mathcal{S}) } \cfrac{N}{R(t)} = \sum_{t =1}^{n_u(\mathcal{S})} \mathcal{X}_u(t) \cfrac{N}{R(t)}.
\end{align}

Let $\mathcal{Y}_u(t)$ be a Bernoulli random variable whose argument is $\epsilon_u(R(t),R_u(t))$. It can easily be seen that $\mathcal{X}_u(t)$ are realizations of $(\mathcal{Y}_u(t)|R(t)\leq R_u(t))$. Let $\overline{\epsilon}_u = \mathds{E}_{R(t)\leq R_u(t)}[\epsilon_u(R(t),R_u(t))]$ and let $\overline{\mathcal{Y}}_u$ be a Bernoulli random variable whose argument is $\overline{\epsilon}_u$. This paper approximate $\mathcal{X}_u(t)$ as realizations of $\overline{\mathcal{Y}}_u$. The rational behind such approximation is that it hold for any $\epsilon_u(R(t),R_u(t))$ equal to a constant for $R(t)\leq R_u(t)$. Such approximation allows to write the average of $\mathcal{E}_u(\mathcal{S})$ as follows:
\begin{align}
\mathds{E}[\mathcal{E}_u(\mathcal{S})] &= \sum_{t =1}^{n_u(\mathcal{S})} \mathds{E}[\mathcal{X}_u(t)] \cfrac{N}{R(t)} \nonumber \\
&= \sum_{t =1}^{n_u(\mathcal{S})} \overline{\epsilon}_u \cfrac{N}{R(t)} \nonumber \\
&= \overline{\epsilon}_u \sum_{t =1}^{n_u(\mathcal{S})} \cfrac{N}{R(t)}.
\label{eql17}
\end{align}
By definition of $\mathcal{C}_u(\mathcal{S})$ in \eref{ctu}, we have $\mathds{E}[\mathcal{E}_u(\mathcal{S})] = \overline{\epsilon}_u\mathcal{C}_u(\mathcal{S})$. Finally, using the weak law of large number, we approximate the quantity $\mathcal{E}_u(S)$ by its mean $\mathds{E}[\mathcal{E}_u(\mathcal{S})]$. Substituting \eref{eql17} in \eref{eql16} and rearranging the terms yields the final expression:
\begin{align}
\mathcal{C}_u(\mathcal{S}) \approx \left( \cfrac{N|\mathcal{F}|}{\tilde{R}_u(\mathcal{S})} + \mathcal{T}_u(S) \right) \left( 1 - \overline{\epsilon}_u \right)^{-1}
\end{align}

\section{Proof of \lref{l1}} \label{app2}

Let $\mathcal{K}_{R(t)}$ be the set of users that can potentially increase the anticipated overall completion time in the $t$-th transmission at the rate $R(t)$. The mathematical expression of the set is:
\begin{align}
\mathcal{K}_{R(t)} = \left\{ u \in \mathcal{U} \ | \ \mathcal{C}_u(t) \geq \mathcal{C}_{u^*}(t-1) \right\}.
\end{align}
Assume that the $t$-th transmission at the rate $R(t)$ is erased at the $u$-th user, then the harmonic mean of the successfully received instantly decodable transmission is unchanged, i.e., $\tilde{R}_u(t-1) = \tilde{R}_u(t)$. The relationship linking the average transmission erasures $\tilde{\epsilon}(t-1)$ and $\tilde{\epsilon}(t)$ of the $u$-th user is the following:
\begin{align}
\tilde{\epsilon}(t) &= \frac{1}{t} \sum_{k=1}^t \epsilon(R(j),R_u(j)) \nonumber \\
&= \frac{t-1}{t} \frac{1}{t-1} \sum_{k=1}^{t-1} \epsilon(R(j),R_u(j)) + \frac{\epsilon(R(t),R_u(t))}{t} \nonumber \\
&= \frac{t-1}{t}\tilde{\epsilon}(t-1) + \frac{\epsilon(R(t),R_u(t))}{t}
\label{era}
\end{align}
For a large number of transmissions, the quantity $\frac{t-1}{t}$ can be approximated be $1$ and as the erasure $0 \leq \epsilon(R(t),R_u(t)) \leq 1,\ \forall\ t>0$, then the $\frac{\epsilon(R(t),R_u(t))}{t} \approx 0$ for large $t$. Therefore, the average erasure probability at the $t$-th transmission is approximatively equal to $\tilde{\epsilon}(t-1)$. Therefore, if the transmission is erased at the $u$-th user, its anticipated completion time does not change, i.e., $\mathcal{C}_u(t) = \mathcal{C}_{u}(t-1)$. In other words, such user is unable to increase the maximum anticipated completion time from its currently maximal value.

Now assume that the $t$-th transmission is successfully received at the $u$-th user. Two scenarios can be distinguished depending on the instant decodability of the transmission. It is clear that if the transmission is non-instantly decodable for a user $u$ that user, then the average rate of instantly decodable transmission is unchanged, i.e., $\tilde{R}_u(t-1) = \tilde{R}_u(t)$. Therefore, using the approximation of the erasure \eref{era}, the anticipated completion time of user $u$ can be written as:
\begin{align}
\mathcal{C}_u(t) = \mathcal{C}_u(t-1) + N/R(t)
\end{align}
Now assume that the transmission is instantly decodable and successfully received. If the transmission is instantly decodable, then the relationship linking $\tilde{R}(t)$ and $\tilde{R}(t-1)$ is the following:
\begin{align}
\cfrac{1}{\tilde{R}(t)} = \cfrac{n-1}{n\tilde{R}(t-1)} + \cfrac{1}{nR},
\label{eq1}
\end{align}
where $n$ is the number of instantly decodable transmissions received, so far, by user $u$. We approximate $\tilde{R}(t)$ by $\tilde{R}(t-1)$ as it indeed holds for large number of messages $\mathcal{F}$ as shown in \eref{eq1}. Assuming that $\tilde{R}_u(t-1) \approx \tilde{R}_u(t)$ and $\tilde{\epsilon}(t-1) \approx \tilde{\epsilon}(t)$, the anticipated completion time of user $u$ remains unchanged, i.e., $\mathcal{C}_u(t) = \mathcal{C}_u(t-1)$.

According to our previous analysis, the individual completion time of user $u$ after a successfully received transmission of the message combination $\kappa(t)$ with the rate $R(t)$ can be written as:
\begin{align}
\mathcal{C}_u(t) =
\begin{cases}
\mathcal{C}_u(t-1) \hspace{0.5cm} &\text{if } u \in \tau_{\kappa(t)} \\
\mathcal{C}_u(t-1) + N/R(t) \hspace{0.5cm} &\text{if } u \notin \tau_{\kappa(t)}.
\end{cases}
\end{align}
Therefore, the set of users that can potentially increase the maximum individual completion time can be reformulated as:
\begin{align}
\mathcal{K}_{R(t)} = \left\{ u \in \mathcal{U} ~\middle|~ \ \mathcal{C}_u(t-1) + \frac{N}{R(t)} \geq \mathcal{C}_{u^*}(t-1) \right\}.
\end{align}

Let $\mathbf{A}$ be the event that the maximum anticipated completion time increases. The probability of such event can be expressed as follows:
\begin{align}
\mathds{P}(\mathbf{A}) &= \mathds{P}(\max_{ u \in \mathcal{U}} (\mathcal{C}_u(t)) > \max_{ u \in \mathcal{U}}(\mathcal{C}_u(t-1))) \nonumber \\
&= 1 - \mathds{P}(\max_{ u \in \mathcal{U}} (\mathcal{C}_u(t)) = \mathcal{C}_{u^*}(t-1)).
\end{align}

By the construction of the set $\mathcal{K}_{R(t)}$, a user $u \notin \mathcal{K}(t)$ is unable to increase the maximum individual completion time in the $t$-th transmission at rate $R(t)$. Hence, for a transmission at rate $R(t)$, the probability of the event $\mathbf{A}$ is:
\begin{align}
\mathds{P}(\mathbf{A}) &= 1 - \mathds{P}(\max_{ u \in \mathcal{K}_{R(t)}} (\mathcal{C}_u(t)) = \mathcal{C}_{u^*}(t-1)) \nonumber \\
&= 1 - \prod_{u \in \mathcal{K}_{R(t)}}\mathds{P}( \mathcal{C}_u(t) - \mathcal{C}_u(t-1) = 0)
\end{align}
From the previous analysis, the anticipated completion time for a user to whom the transmission is instantly decodable is $0$ regardless is the transmission is actually successfully received or erased. However, for users to whom the transmission is not instantly decodable, the increase in the anticipated completion time occurs if and only if the transmission is successfully received since $\tilde{\epsilon}(t-1) \approx \tilde{\epsilon}(t)$. Therefore, due to the dynamic nature of the channel realization, from the previous analysis, the increase in the anticipated completion time for user $u \notin \tau_{\kappa(t)}$ with non-empty Wants set, is a Bernoulli random variable that takes the following values:
\begin{align}
&\mathcal{C}_u(t) - \mathcal{C}_u(t-1) = \nonumber \\
& \qquad \begin{cases}
0 &\text{with probability } \epsilon(R(t),R_u(t)) \\
N/R(t) &\text{with probability } 1-\epsilon(R(t),R_u(t))
\end{cases}
\end{align}
Therefore, the probability that the maximum anticipated completion time increases at the $t$-th transmission as compared with the $t-1$-th transmission can be expressed as follows:
\begin{align}
\mathds{P}(\mathbf{A}) &= 1 - \prod_{u \in (\mathcal{K}_{R(t)} \cap \mathcal{N}) \setminus \tau_{\kappa}} \epsilon(R(t),R_u(t)),
\end{align}
where $\mathcal{N}=\{u \in \mathcal{U} \ | \ \mathcal{W}_u \neq \varnothing\}$ is the set of users with non-empty Wants set. Finally, the message combination $k(t)$ and the transmission $R(t)$ that minimizes the probability of increase in the anticipated completion time for the $t$-th transmission can be expressed as follows:
\begin{align}
&\min_{ \substack{\kappa(t) \in \mathcal{P}(\mathcal{F}) \\ R(t) \in \mathcal{R}(t)}} 1 - \prod_{u \in (\mathcal{K}_{R(t)} \cap \mathcal{N}) \setminus \tau_{\kappa(t)}} \epsilon(R(t),R_u(t)) \nonumber \\
&= \max_{ \substack{\kappa(t) \in \mathcal{P}(\mathcal{F}) \\ R(t) \in \mathcal{R}(t)}} \prod_{u \in (\mathcal{K}_{R(t)} \cap \mathcal{N}) \setminus \tau_{\kappa(t)}} \epsilon(R(t),R_u(t)).
\end{align}
Therefore, the expected increase in the completion time reduction problem can be approximated by the following joint optimization over the message combination $\kappa(t)$ and the transmission rate $R(t)$:
\begin{align}
&(\kappa^*(t), R^*(t)) \\
& =\arg \max_{ \substack{\kappa(t) \in \mathcal{P}(\mathcal{F}) \\ R(t) \in \mathcal{R}(t)}} \prod_{u \in (\mathcal{K}_{R(t)} \cap \mathcal{N}) \setminus \tau_{\kappa(t)}} \frac{\epsilon(R(t),R_u(t))N}{R(t)} \nonumber \\
& = \arg \max_{ \substack{\kappa(t) \in \mathcal{P}(\mathcal{F}) \\ R(t) \in \mathcal{R}(t)}} \sum_{u \in (\mathcal{K}_{R(t)} \cap \mathcal{N}) \setminus \tau_{\kappa(t)}} \log (\frac{\epsilon(R(t),R_u(t))N}{R(t)}) \nonumber \\
& = \arg \min_{ \substack{\kappa(t) \in \mathcal{P}(\mathcal{F}) \\ R(t) \in \mathcal{R}(t)}} \sum_{u \in (\mathcal{K}_{R(t)} \cap \tau_{\kappa(t)})} \log (\frac{\epsilon(R(t),R_u(t))N}{R(t)}) \nonumber \\
& = \arg \max_{ \substack{\kappa(t) \in \mathcal{P}(\mathcal{F}) \\ R(t) \in \mathcal{R}((t))}} \sum_{u \in \mathcal{K}_{R(t)}\cap \tau_{\kappa(t)}} \log (\frac{R(t)}{\epsilon(R(t),R_u(t))N}).\nonumber
\end{align}

\section{Proof of \thref{th2}} \label{app3}

To prove this theorem, this section first shows that there exists a one-to-one mapping between the set of feasible message combination, transmission rate, and targeted users and the set of maximal cliques in the RA-IDNC graph. To conclude the proof, it is sufficient to show that the weight of the cliques is equivalent to the objective function of the optimization problem \eref{opti} under investigation. Therefore, this section proposes first to show that for any maximal clique, there exist a unique combination message combination, transmission rate, and targeted users. Afterward, the converse is shown by proving that every feasible message combination, transmission rate, and targeted users is represented by a unique maximal clique in the RA-IDNC graph.

Let the side information matrix $\mathbf{S}=[S_{uf}]$ and the individual side information $S_u$ of user $u$ be defined as in \appref{app1}. Let $\mathbf{M}$ be a maximal clique in the RA-IDNC graph. We show that the transmission of the message mix $\kappa = \bigoplus\limits_{v_{ufr} \in \mathbf{M}} f$ at the transmission rate $R=r, v_{ufr} \in \mathbf{M}$ is instantly decodable for the users $\tau_\kappa = \left\{ u \right\}_{v_{ufr} \in \mathbf{M}}$. First note that the combination of message combination, rate and users is well-defined. The messages $f \in \mathcal{F}, \ \forall \ v_{ufr} \in \mathbf{M}$, thus they are combinable with the XOR combination to produce an IDNC-coded message. Moreover, $u \in \mathcal{U}, v_{ufr} \in \mathbf{M}$ by the construction of the vertices. Further, since $\mathbf{M}$ is a clique then all the vertices are connected, i.e., they satisfy C1. In other words, $R=r=r^\prime, , \ \forall \ v_{ufr},v_{u^\prime f^\prime r^\prime } \in \mathbf{M}$.

To show that the transmission represented by the clique $\mathbf{M}$ is a feasible transmission, we only need to demonstrate that the message mix $\kappa$ transmitted at the rate $R$ is instantly decodable for all users $u \in \tau_{\kappa}$. First define the set $\alpha$ of user that can instantly decode the transmission $(\kappa,R)$ as follows
\begin{align}
\alpha(\kappa,R) = \{ u \in \mathcal{U} \ | \ S_u^T k = 1 \text{ and } R \leq R_u\} ,
\label{eqth11}
\end{align}
where the notation $X^T$ refers to the transpose of vector $X$. We want to show that $\alpha(\kappa,R) = \tau_{\kappa}$. We begin by showing the first inclusion ($\tau_{\kappa} \in \alpha(\kappa,R)$). Afterwards, we show the second one.

Let $u \in \tau_{\kappa}$ then it is clear that $R \leq R_u$ by the construction of the vertices. It is also clear to note, by the construction of the vertices, that $S_u^T k \neq 0$ since each vertex $v_{ufr}$ contains a message $f \in \mathcal{W}_u$. Assume that $S_u^T k > 1$, therefore $\exists \ f \neq f^\prime$ such that $v_{ufr}$ and $v_{uf^\prime r}$ are two vertices in $\mathbf{M}$. Since the two vertices represents the same user but for different messages then they do not satisfy C2 and thus are not connected. However, $\mathbf{M}$ is a clique and thus all the vertices which conclude that $S_u^T k = 1$ and finally that $\tau_{\kappa} \subset \alpha(\kappa,R)$.

We now prove the other inclusion. Assume $\exists \ u \in \alpha(\kappa,R)$ such that $u \notin \tau_{\kappa}$. Let $f \in \mathcal{F}$ that ensures $S_u^T k = 1$. The vertex $v_{ufR}$ satisfy condition C1 with all the vertices $v \in \mathbf{M}$. Further, since the transmission $(\kappa,R)$ is instantly decodable for all users $u$ in $\tau_{\kappa}$ then the message combination $f$ satisfy the condition C2 with all the vertices $v \in \mathbf{M}$. Therefore, vertex $v_{ufR}$ is connected to all the vertices in the clique $\mathbf{M}$ which proves that $v_{ufR} \cup \mathbf{M}$ is also a clique. However, since $\mathbf{M}$ is a maximal clique then $v_{ufR} \cup \mathbf{M}$ cannot be a clique. Finally $\alpha(\kappa,R) \subset \tau_{\kappa}$ and we obtain $\alpha(\kappa,R) = \tau_{\kappa}$. As a conclusion each maximal clique $\mathbf{M}$ represents a feasible transmission. The uniqueness flows directly from the uniqueness of the message combination $\kappa$ and the transmission rate $R$.

We now show the converse, i.e., each feasible transmission can be uniquely represented by a clique in the RA-IDNC graph. Let $\kappa$ be a message combination and $R$ a transmission rate. Define the set of users $\alpha(\kappa,R)$ that can instantly decode the transmission $(\kappa,R)$ in a similar way as in \eref{eqth11}. Let $f(u)$ be the message intended to user $u \in \alpha(\kappa,R)$. We want to show that the set of vertices $\left\{ v_{u f(u) R}\right\}_{u \in \alpha(\kappa,R)} = \mathbf{M}$ form a maximal clique in the RA-IDNC graph.

Let $v_{u f R}$ and $v_{u^\prime f^\prime R}$ be two vertices from the set $\mathbf{M}$ defined earlier. These two vertices verify the connectivity condition C1 since the rate is the same. If $f = f^\prime$ then it is clear that the vertices satisfy condition C2 and thus they are connected. now assume that $f \neq f^\prime$ and assume either $f^\prime \notin \mathcal{H}_u$ or $f \notin \mathcal{H}_{u^\prime}$. Therefore, $f^\prime \in \mathcal{W}_u$ or $f \in \mathcal{W}_{u^\prime}$ which leads to $S_u^T k > 1$ or $S_{u^\prime}^T k > 1$ since $\kappa_f = \kappa_{f^\prime} = 1$ and $S_{uf} = S_{uf^\prime} = S_{u^\prime f} = S_{u^\prime f^\prime}=1$. This results in $u$ or $u^\prime \notin \alpha(\kappa,R)$ in contradiction with the first assumption. Finally, we conclude that vertices $v_{u f R}$ and $v_{u^\prime f^\prime R}$ satisfy condition C2 and thus they are connected. Therefore $\mathbf{M}$ is a clique. Showing that $\mathbf{M}$ is a maximal clique follows the same steps as proving that $\alpha(\kappa,R) = \tau_{\kappa}$ and the uniqueness of the clique follows from the uniqueness of $\alpha(\kappa,R)$. Hence, there a one-one mapping between the set of maximal cliques in the RA-IDNC graph and the set of feasible schedule.

To conclude the proof, we show that the weight of the clique is the objective function of \eref{opti}. Let the weight of vertex $v_{ufr}$ be defined as in \eref{weight}. The weight of a clique $\mathbf{M}$ with its corresponding message mix $\kappa$ and transmission rate $R$ is:
\begin{align}
w(\mathbf{M}) &= \sum_{v\in \mathbf{M}} w(v) = \sum_{u \in \tau_{\kappa}} w(v_{ufR}) \nonumber \\
&= \sum_{u \in \tau_{\kappa} \cap \mathcal{K}_R} \log (\cfrac{R}{\epsilon_u(R,R_u)N}) .
\end{align}

Therefore, the problem of reducing the completion time \eref{opti} in RA-IDNC-based networks is equivalent to the maximum weight clique problem among the maximal cliques in the RA-IDNC graph.

\bibliographystyle{IEEEtran}
\bibliography{citations}

\end{document}